\documentclass[11pt]{article}
\pdfoutput=1
\usepackage[colorlinks,pagebackref=true]{hyperref}
\usepackage[utf8]{inputenc}
\usepackage[T1]{fontenc}
\usepackage[margin=3cm]{geometry}
\usepackage{microtype}
\usepackage[english]{babel}
\usepackage{mathtools, amsthm, amsfonts, amssymb}
\usepackage{enumerate}
\usepackage{booktabs}
\usepackage{accents}

\hypersetup{
  linkcolor=[rgb]{0.3,0.3,0.6},
  citecolor=[rgb]{0.2, 0.6, 0.2},
  urlcolor=[rgb]{0.6, 0.2, 0.2}
}

\theoremstyle{plain}
\newtheorem{theorem}{Theorem}[section]
\newtheorem{lemma}[theorem]{Lemma}

\newtheorem{corollary}[theorem]{Corollary}
\newtheorem{proposition}[theorem]{Proposition}

\newtheorem*{theorem*}{Theorem}
\newtheorem*{proposition*}{Proposition}

\theoremstyle{definition}
\newtheorem{definition}[theorem]{Definition}

\newtheorem{example}[theorem]{Example}

\newtheorem*{definition*}{Definition}

\newcommand{\smallo}{o} %

\newcommand*{\cH}{\mathcal{H}}

\newcommand*{\cP}{\mathcal{P}}

\newcommand*{\cS}{\mathcal{S}}

\newcommand*{\RR}{\mathbb{R}}

\newcommand*{\CC}{\mathbb{C}}
\newcommand*{\FF}{\mathbb{F}}
\newcommand*{\NN}{\mathbb{N}}

\newcommand*{\GL}{\mathrm{GL}}

\newcommand*{\flatten}{\mathrm{flatten}}

\newcommand*{\eps}{\varepsilon}

\DeclareMathOperator{\rank}{R}

\newcommand*{\spec}{\mathrm{spec}}
\newcommand*{\supp}{\mathrm{supp}}
\newcommand*{\tr}{\mathrm{tr}}

\newcommand*{\Herm}{\mathrm{Herm}}

\newcommand{\ra}{\rightarrow}

\DeclarePairedDelimiterX{\inner}[2]{\langle}{\rangle}{#1, #2}
\DeclareMathOperator{\subrank}{Q}
\DeclareMathOperator{\trank}{R}
\DeclareMathOperator{\asympsubrank}{\wtilde{Q}}
\DeclareMathOperator{\slicerank}{SR}

\DeclareMathOperator{\asymprank}{\wtilde{R}}

\DeclareMathAccent{\wtilde}{\mathord}{largesymbols}{"65}

\newcommand{\asympleq}{\lesssim}
\newcommand{\asympgeq}{\gtrsim}
\newcommand{\symleq}{\leq_{\mathrm{s}}}
\newcommand{\symgeq}{\geq_{\mathrm{s}}}
\newcommand{\symasympleq}{\lesssim_{\mathrm{s}}}
\newcommand{\symasympgeq}{\gtrsim_{\mathrm{s}}}
\DeclareMathOperator{\symasymprank}{\wtilde{R}_{s}}
\DeclareMathOperator{\symasympsubrank}{\wtilde{Q}_{s}}
\DeclareMathOperator{\symsubrank}{Q_s}
\DeclareMathOperator{\symrank}{R_s}
 
\newcommand{\Id}{\mathrm{Id}}

\title{Symmetric Subrank of Tensors and Applications}

\author{Matthias Christandl\footnote{Department of Mathematical Sciences, University of Copenhagen, \href{mailto:christandl@math.ku.dk}{christandl@math.ku.dk}} \and Omar Fawzi\footnote{Univ.~Lyon, ENS Lyon, UCBL, CNRS, Inria, LIP, \href{mailto:omar.fawzi@ens-lyon.fr}{omar.fawzi@ens-lyon.fr}} \and Hoang Ta\footnote{Univ.~Lyon, ENS Lyon, UCBL, CNRS, Inria, LIP, \href{mailto:duy-hoang.ta@ens-lyon.fr}{duy-hoang.ta@ens-lyon.fr}} \and Jeroen Zuiddam\footnote{Korteweg-de Vries Institute for Mathematics, University of Amsterdam, \href{mailto:j.zuiddam@uva.nl}{j.zuiddam@uva.nl}}}

\begin{document}
\maketitle

\begin{abstract}
	Strassen (Strassen, J.~Reine Angew. Math., 375/376, 1987) introduced the subrank of a tensor as a natural extension of matrix rank to tensors. Subrank measures the largest diagonal tensor that can be obtained by applying linear operations to the different indices (legs) of the tensor (just like the matrix rank measures the largest diagonal matrix that can be obtained using row and column operations).
    Motivated by problems in combinatorics and complexity theory we introduce the new notion of \emph{symmetric subrank} of tensors by restricting these linear operations to be the same for each index.

	We prove precise relations and separations between subrank and symmetric subrank.
	We prove that for symmetric tensors 
	the subrank and the symmetric subrank are asymptotically equal. This proves the asymptotic subrank analogon of a conjecture known as Comon's conjecture in the theory of tensors.
	This result allows us to prove a strong connection between the general and symmetric version of an asymptotic duality theorem of Strassen.
	We introduce a representation-theoretic method to asymptotically bound the symmetric subrank called the \emph{symmetric quantum functional} in analogy with the quantum functionals (Christandl, Vrana, Zuiddam, J.~Amer.~Math.~Soc., 2021), and we study the relations between these functionals. 
\end{abstract}
\vspace{1em}
\newpage
\tableofcontents
\newpage

\section{Introduction}

Symmetry is a central theme in the theory of tensors \cite{CGLM08,Shi18}. We study how symmetry influences the tensor parameter \emph{subrank} which is closely connected to problems in combinatorics, quantum information theory and algebraic complexity theory.
The subrank of a tensor is a natural extension of matrix rank to tensors that measures the largest diagonal tensor that can be obtained by applying linear operations to the different indices (i.e.~generalizations of rows and columns of a matrix) of the tensor.
This parameter was introduced by Strassen \cite{Str86, Strassen1987RelativeBC,Str88,Str91} as a method to study fast matrix multiplication algorithms (for an introduction to the research on fast matrix multiplication algorithms we refer to \cite{PMM97,blaser2013fast}), and has since been studied from several points of view, including quantum information theory \cite{MR3390878, CVZ18}, algebraic geometry \cite{MR3966415,kopparty_et_al, DBLP:journals/corr/abs-2012-04679}, combinatorics and communication complexity \cite{DBLP:journals/corr/abs-2111-08262}.

We introduce and analyse the \emph{symmetric subrank} of tensors, which we define as the largest diagonal tensor that can be obtained from a tensor by applying \emph{the same} linear operation to all dimensions.

\subsection{Motivation: Independent set problems in combinatorics}

Various important problems in combinatorics are a special case of the problem of determining the independence number of a hypergraph, which is the size of the largest subset of the vertex set that does not induce any edges. Of particular interest is the rate of growth of the independence number under taking strong powers of a fixed hypergraph, which in the case of undirected graphs is called the Shannon capacity \cite{sha56}.

The subrank, as an algebraic relaxation of the independence number, provides a natural method to upper bound the independence number. Indeed the independence number is upper bounded by the subrank of any tensor that ``fits'' the hypergraph (i.e.~has the appropriate support), in a fashion that is similar to the Haemers bound in graph theory~\cite{HW79}. Several more tensor methods have been introduced in this context to find good independence number upper bounds, notably the slice rank~\cite{TS16}, analytic rank~\cite{Lovett} and related parameters. 

It has been realized, however, that, while these methods have been very successful in solving various open problems in combinatorics, they all suffer from a barrier that renders them useless in the case where the independence number is low but the tensors fitting the hypergraph have large induced matchings in their support. This is the case for example for the corner problem~\cite{DBLP:journals/corr/abs-2111-08262}, and improving the current bound for the capset problem requires going beyond this barrier~\cite{EG17,Matrix_capset}. This ``induced matching barrier'' calls for an effort of finding methods for upper bounding the independence number that can go below the induced matching number.

\subsection{From subrank to symmetric subrank}

As part of the fundamental study of symmetry in tensor theory (and aiming to circumvent the aforementioned induced matching barrier), we introduce the symmetric subrank of tensors. Whereas the subrank~$\subrank(f)$ of a tensor $f\in \FF^{n_1} \otimes \cdots \otimes \FF^{n_k}$ (over a field $\FF$) measures the largest number~$r$ such that the diagonal tensor $\langle r\rangle = \sum_{i=1}^r e_i \otimes \cdots \otimes e_i \in \FF^r \otimes \cdots \otimes \FF^r$ (where the~$e_i$ form the standard basis of $\FF^r$) can be obtained from $f$ by acting with linear operations $A^{(i)} : \FF^{d_i} \to \FF^r$ on $f$, that is 
\[
	\langle r\rangle = (A^{(1)} \otimes \cdots \otimes A^{(k)}) f,
\]
the \emph{symmetric subrank} $\symsubrank(f)$ of a tensor $f \in \FF^d \otimes \cdots \otimes \FF^d$ we define as the largest number $r$ such that there is a linear map $A : \FF^d \to \FF^r$ so that 
\[
	\langle r \rangle = (A \otimes \cdots \otimes A) f.
\]
The symmetric subrank is not just defined on symmetric tensors, but on all tensors.
For every tensor $f$ we have $\symsubrank(f) \leq \subrank(f) \leq d$. On the applications side, the symmetric subrank, as we will prove, still upper bounds the independence number (Proposition~\ref{adj_boundsubrank}), but can be strictly smaller than the subrank (see e.g., Example~\ref{ex:cycles}). Sections~\ref{subsec:matrices} and \ref{subsec:symm-complex-matrices} are devoted to tensors of order two, i.e., matrices, and we show in particular that for symmetric matrices, the symmetric subrank and the subrank are equal (Theorem~\ref{symmetric_matrix}). For tensor of order $k \geq 3$, we show that for symmetric tensor, whenever the subrank takes the maximal value $d$, then so does the symmetric subrank (Theorem~\ref{thm:snap-to-max}). %

\subsection{Asymptotic symmetric subrank and asymptotic spectrum duality}

Central in this paper (Section~\ref{sec:asymp-symm-subrank}) is the study of the asymptotic behaviour of the subrank and symmetric subrank. This is captured by the asymptotic subrank
\[
\asympsubrank(f) = \lim_{n\to\infty} \subrank(f^{\otimes n})^{1/n}
\]
and the asymptotic symmetric subrank
\[
	\symasympsubrank(f) = \lim_{n\to\infty} \symsubrank(f^{\otimes n})^{1/n}.\footnote{For this definition to make sense (i.e.~the limit to exist) we need to put mild conditions on the tensor, see Propoposition~\ref{limsup_sup}. However, we can give a more general definition by replacing the lim by a limsup which is always valid.}
\]
This notion is useful to bound the rate of growth of the independence number of powers of hypergraphs, for instance. We show that for tensors $f$ (not necessarily symmetric) of order two, we have $\symasympsubrank(f) = \asympsubrank(f)$ (Theorem~\ref{thm:matrix-asymp}). For tensors of arbitrary order, we expect we can have $\symasympsubrank(f) < \asympsubrank(f)$. However, we show that for symmetric tensors $f$ we also have equality $\symasympsubrank(f) = \asympsubrank(f)$ (Theorem~\ref{th:asympsubrank}). This result can be interpreted as saying that Comon's conjecture is true \emph{asymptotically} for the subrank. This is discussed further in Section~\ref{subsec:rel-asymp} where we show more generally that Comon's conjecture holds for the asymptotic restriction pre-order.

Strassen \cite{Str86,Strassen1987RelativeBC,Str88,Str91} proved a strong duality theorem that describes the asymptotic subrank (which can naturally be thought of as a ``maximization problem'') as a minimization problem:
\[
	\asympsubrank(f) = \min_{\phi \in X} \phi(f),
\] 
where the ``dual space'' $X$ is the asymptotic spectrum of tensors, a set of very special, well-behaved tensor parameters. For background we refer to the recent works \cite{CVZ18,Zui18_thesis,WZ}.
We introduce in Section~\ref{sec:asympspec} the asymptotic spectrum of symmetric tensors~$X_\mathrm{s}$ and prove the analogous duality theorem for the asymptotic symmetric subrank $\symasympsubrank$ of symmetric tensors. Moreover, in Section~\ref{sec:symmetricquantum}, we construct an explicit point in the dual space~$X_\mathrm{s}$ called the symmetric quantum functional, based on a construction in \cite{CVZ18}.

\subsection{Related work}

The general question of upper bounds on the Shannon capacity of hypergraphs is particularly well-studied in the special setting of undirected graphs, from which the name ``Shannon capacity'' comes: it in fact corresponds to the zero-error capacity of a channel~\cite{sha56}. 
Even for undirected graphs, it is not clear how to compute the Shannon capacity in general, but some methods were developed to give upper bounds. 
The difficulty is to find a good upper bound on the largest independent set that behaves well under the product $\boxtimes$. 
For undirected graphs, the best known methods are the Lov\'asz theta function~\cite{Lo79}, and the Haemers bound which is based on the matrix rank~\cite{HW79}. 

For hypergraphs, we only know of algebraic methods that are based on various notions of tensor rank, and in particular the slice rank~\cite{TS16}, and similar notions like the analytic rank \cite{gowers2011linear, Lovett}, the geometric rank \cite{kopparty_et_al}, and the G-stable rank \cite{derksen2020gstable}. Even though the slice rank is not multiplicative under $\boxtimes$ it is possible to give good upper bounds on the asymptotic slice rank via an asymptotic analysis \cite{TS16}, which is closely related to the Strassen support functionals~\cite{Str91} or the more recent quantum functionals~\cite{CVZ18}.

\section{Symmetric subrank}
\label{sec:symmetric_subrank}

In this section we define the new notion of symmetric subrank and then discuss its basic properties and separations. 

\subsection{Symmetric subrank}\label{subsec:symm-subrank-def}
We will first recall the definition of the subrank. Then we will give the definition of the symmetric subrank. After that we will discuss the connection to the independence number of hypergraphs.

\begin{definition}[Restriction, unit tensor and subrank]
For two tensors $f \in \FF^{d_1} \otimes \cdots \otimes \FF^{d_k}$ and $g \in \FF^{e_1} \otimes \cdots \otimes \FF^{e_k}$ we write $g \leq f$ if there are linear maps $A^{(i)} : \FF^{d_i} \to \FF^{e_i}$ such that $g = (A^{(1)} \otimes \cdots \otimes A^{(k)})f$. If $g \leq f$ then we say $g$ is a restriction of $f$, and we call $\leq$ the \emph{restriction} order on tensors.

For any nonnegative integer $r$ we define the \emph{unit tensor} $\langle r\rangle = \sum_{i=1}^r e_i^{\otimes k} \in (\FF^r)^{\otimes k}$ where the $e_i$ denote the standard basis elements of $\FF^r$.

For any tensor $f \in \FF^{d_1} \otimes \cdots \otimes \FF^{d_k}$ the \emph{subrank} $\subrank(f)$ is defined as the largest number $r$ such that~$\langle r\rangle \leq f$.
\end{definition}

We note that for tensors $f \in \FF^{d_1} \otimes \FF^{d_2}$ (matrices) the subrank equals the usual rank of the matrix. This follows from the fact that any matrix can be put in diagonal form by invertible row and column operations (Gaussian elimination) so that there are $\rank(f)$ nonzero elements on the diagonal.

The \emph{symmetric} subrank of a tensor (we do not require the tensor to be symmetric) is defined in the same way as the subrank with the extra requirement that all linear maps $A^{(i)}$ are the same.
  
\begin{definition}[Symmetric restriction and symmetric subrank] \label{symmetricsubrank}
	For any two (not necessarily symmetric) tensors $f \in (\FF^d)^{\otimes k}$ and $g \in (\FF^{e})^{\otimes k}$
	we write
	$g \symleq f$ if there is a linear map $A: \FF^{d} \ra \FF^{e}$ such that $g = A^{\otimes k} f$. We call ${\symleq}$ the \emph{symmetric restriction} order on tensors.

	For any tensor $f \in (\FF^d)^{\otimes k}$ the \emph{symmetric subrank} $\symsubrank(f)$ is defined as the largest number~$r$ such that $\langle r\rangle \leq_s f$.
\end{definition}

A natural ``dual'' of the symmetric subrank called \emph{symmetric rank} is well-studied~\cite{CGLM08}.
The symmetric rank $\symrank(f)$ of a symmetric tensor $f$ is defined as the smallest number $r$ such that $f \symleq \langle r\rangle$. %
In other words, it is the smallest number~$r$ such that there are $r$ vectors $v_i$ so that $f  = \sum_{i=1}^{s}v_{i}^{\otimes k}$. 
In the study of homogeneous polynomials (which correspond naturally to symmetric tensors) this notion is often called Waring rank.

\subsection{Application: independent sets and induced matching barrier}\label{subsec:ind-set}

In this section we discuss some combinatorial background and motivation and in particular discuss the independence number of directed hypergraphs, how the symmetric subrank upper bounds it, and how the symmetric subrank circumvents an ``induced matching barrier'' in this context.

We recall that a directed $k$-uniform hypergraph~$H$ is a pair $(V,E)$ where $V$ is a finite vertex set and $E$ is a set of $k$-tuples of elements in~$V$. A~subset $S \subseteq V$ is an \emph{independent set} of $H$ if $E \cap S^{\times k}$ is empty. The \emph{independence number} $\alpha(H)$ is the size of the largest independent set in~$H$. The support of a tensor $f \in (\FF^n)^{\otimes k}$ is defined as $\supp(f) = \{(i_1, \ldots, i_k) \in [n]^k : f_{i_1, \ldots, i_k} \neq 0\}$ and $([n], \supp(f))$ is an example of a directed $k$-uniform hypergraph. On the other hand, if $H = (V,E)$ is a $k$-uniform directed hypergraph in $n$ vertices, then we define the \emph{adjacency tensor} $A_H \in (\FF^n)^{\otimes k}$ by setting $(A_H)_{i_1, \ldots, i_k} = 1$ if $i_1 = \cdots = i_k$ or $(i_1, \ldots, i_k) \in E$ and setting all other entries of $A_H$ to 0.
\begin{proposition}
	\label{adj_boundsubrank}
	Let $H = (V,E)$ be a directed $k$-uniform hypergraph with $n$ vertices. %
	Let $\FF$ be a field. Let $f \in (\FF^n)^{\otimes k}$ be a tensor such that, for every $e \in [n]^k$ if $e \not\in E$, then $f_{e_1, \ldots, e_k} = 0$, and for every $i \in [n]$, $f_{i,\dots,i} = 1$. Then
	$\alpha(H) \leq \symsubrank(f)$.
\end{proposition}
\begin{proof}
	An independent set in $H$ of size $r$ directly gives $\langle r\rangle \symleq f$. Namely, if $S\subseteq V$ is an independent set, then the subtensor of $f$ indexed by $S^{\times k}$ equals $\langle r\rangle$.
\end{proof}

Finally, we discuss induced matchings and how they pose a barrier for tensor methods to upper bound the independence number of hypergraphs. Later we will see that the symmetric subrank does not suffer from this barrier. 

Let $H = (V,E)$ be a directed $k$-uniform hypergraph. Let $\Phi = E \cup \{(v,\ldots, v): v \in V\} \subseteq V \times \cdots \times V$. We say that a subset $M \subseteq \Phi$ is an \emph{induced matching} if the elements in $S$ are disjoint in all coordinates and if $M = E \cap (M_1 \times \cdots \times M_k)$ where $M_i = \{m_i : m\in M\}$. Let $\beta(H)$ be the size of the largest subset $M \subseteq \Phi$ that is an induced matching. Note that $\alpha(H)$ is the size of the largest subset $S \subseteq V$ such that $\{(s, \ldots, s): s \in S\} \subseteq \Phi$ is an induced matching. Therefore, $\alpha(H) \leq \beta(H)$.

\begin{proposition}\label{prop:beta-bound}
	Let $H = (V,E)$ be a directed $k$-uniform hypergraph with $n$ vertices. 
	Let $\FF$ be a field. Let $f \in (\FF^n)^{\otimes k}$ be a tensor such that, for every $e \in [n]^k$ if $e \not\in E$, then $f_{e_1, \ldots, e_k} = 0$, and for every $i \in [n]$, $f_{i,\dots,i} = 1$. Then
	$\beta(H) \leq \subrank(f)$.
\end{proposition}
\begin{proof}
	An induced matching $M \subseteq \Phi = E \cup \{(v,\ldots, v): v \in V\}$ of size $r$ directly gives $\langle r\rangle \leq f$. Namely, the subtensor of $f$ indexed by $M_1 \times \cdots \times M_k$ equals $\langle r\rangle$ up to permuting the coordinates of each factor $\FF^n$.
\end{proof}

We see from Proposition~\ref{adj_boundsubrank} and Proposition~\ref{prop:beta-bound} that, while both the symmetric subrank $\symsubrank(f)$ and the subrank $\subrank(f)$ can be used to upper bound the independence number $\alpha(H)$, the subrank cannot give good bounds when $\beta(H)$ is much larger than $\alpha(H)$. We may thus think of $\beta(H)$ as a barrier for~$\subrank(f)$ to give good upper bounds on $\alpha(H)$. Many other tensors methods (slice rank, partition rank, analytic rank, geometric rank, G-stable rank) are also lower bounded by this barrier $\beta(H)$. We will see that indeed the symmetric subrank $\symsubrank(f)$ can be strictly smaller than $\beta(H)$.

\subsection{Tensors of order two (matrices)}
\label{subsec:matrices}
In Section~\ref{subsec:symm-subrank-def} we introduced the symmetric subrank of tensors. With the motivation in mind of using symmetric subrank as a method to upper bound the independence number of hypergraphs as in Section~\ref{subsec:ind-set}, it is natural to ask whether this method is better than using the subrank itself.
It follows directly from the definition of the symmetric subrank that for any $k$-tensor $f$ we have that $\symsubrank(f) \leq \subrank(f)$. Can this inequality be strict?
In this and the following sections we will discuss relations and separations with the ordinary subrank. We obtain precise results under assumptions about the order, ground field and symmetry of the tensors.

In this section we consider tensors of order two. These we can simply think of as matrices via the identification $\sum_{i,j} f_{ij}\, e_i \otimes e_j \mapsto (f_{ij})_{ij}$. In the language of matrices the restriction order and symmetric restriction order are given as follows. For matrices $f \in \FF^{n_1\times n_2}$ and $g \in \FF^{m_1 \times m_2}$ we have $f \leq g$ if there are matrices $A^{(i)} \in \FF^{n_i \times m_i}$ such that $f = A^{(1)} g (A^{(2)})^T$. For matrices $f \in \FF^{n \times n}$ and $g \in \FF^{m \times m}$ we have $f \symleq g$ if there is a matrix $A \in \FF^{n \times m}$ such that $f = AgA^T$. Note in particular how in this formulation we multiply on the left by $A$ and on the right by the transpose of $A$. When $A$ is invertible and $f = AgA^T$ the matrices $f$ and $g$ are often called \emph{congruent}. However we will allow $A$ to be non-invertible. The (symmetric) subrank of a matrix $f$ is now the largest number $r$ such that the $r \times r$ diagonal matrix $\langle r\rangle$ is a (symmetric) restriction of $f$. 

First of all, as a basic fact that we will use later, we note that for any matrix $f$ the subrank~$\subrank(f)$ equals the usual notion of matrix rank $\rank(f)$.

\begin{lemma}\label{lem:matrix-subrank-equals-rank}
	Let $f$ be a matrix, then $\subrank(f) = \rank(f)$.
\end{lemma}
\begin{proof}
	Clearly $\subrank(f) \leq \rank(f)$. It is well-known that by Gaussian elimination we can find invertible matrices $A^{(i)}$ such that $f$ is a diagonal matrix with $\rank(f)$ nonzero entries. Thus $\subrank(f) \geq \rank(f)$.
\end{proof}

\begin{lemma}\label{lem:skew-sym-zero}
	Let $f$ be a $d\times d$ matrix over an arbitrary field $\FF$ such that $f_{\ell,\ell} = 0$ for all $\ell \in [d]$ and $f_{i,j} = - f_{j,i}$ for all $i\neq j \in [d]$. 
	Then $\symsubrank(f) = 0$.
\end{lemma}
\begin{proof}
	For any matrix $B \in \FF^{m \times d}$ let $g = B f B^T$. Then the diagonal entries $g_{kk}$ are zero for all~$k$. Indeed we have $g_{kk} = \sum_{i,j}B_{ki}B_{kj}f_{ij} = 0$ since $f_{ij} = -f_{ji}$ for all $i \neq j$. We conclude that $\symsubrank(f) = 0$.
\end{proof}

In particular, if $\FF \neq \FF_2$, then the condition in Lemma~\ref{lem:skew-sym-zero} is equivalent to $f = -f^T$, that is, $f$ being skew-symmetric.

\begin{example}\label{example:skew_matrix}
	It is easy to find a $d \times d$ matrix $f$ of full rank that satisfies the condition in Lemma~\ref{lem:skew-sym-zero}. 
	Then by Lemma~\ref{lem:skew-sym-zero} we have $\symsubrank(f) = 0$ while $\subrank(f) = d$. For example, for even $d$ we may take $f$ with entries $f_{i, d+1-i} = 1$ for all $1\leq i \leq d/2$ and $f_{i, d+1-i} = -1$ for all $d/2 < i \leq d$ and all other entries equal to zero, that is,
	\[
		\begin{pmatrix}
		0 & -1 \\
		1 & 0
		\end{pmatrix},\,
		\begin{pmatrix}
			0 & 0 & 0  & -1 \\
			0 & 0 & -1 & 0  \\
			0 & 1 & 0  & 0  \\
			1 & 0 & 0  & 0
		\end{pmatrix},\,\ldots
	\]
\end{example}

\begin{lemma}\label{lem:non-sym-not-max}
	Let $f$ be a non-symmetric $d\times d$ matrix over an arbitrary field $\FF$. Then $\symsubrank(f) < d$.
\end{lemma}
\begin{proof}
	Suppose that $\symsubrank(f) = d$. Then there is a matrix $A$ such that $\langle d\rangle = AfA^T$. Since $\langle d\rangle$ has full rank, $A$ must have full rank. We find that $f = A^{-1}\langle d\rangle (A^T)^{-1} = A^{-1} (A^{-1})^T$ and so $f$ is symmetric. This is a contradiction.
\end{proof}

\begin{example}\label{ex:cycles}
	Let $C_{2k+1}$ be the directed cycle graph with vertex set $\{1,\ldots, 2k+1\}$ and edge set $\{(1,2), (2,3), \ldots, (2k+1, 1)\}$ and let $f$ be the adjacency matrix of $C_{2k+1}$ over any fixed field, so that~$f$ is the $(2k+1) \times (2k+1)$ matrix
	\begin{align*}
		f = 
		\begin{pmatrix}
		1 & 1 & 0 & 0 &\cdots & 0 & 0\\
		0 & 1 & 1 & 0 &\cdots & 0 & 0\\
		\vdots  & \vdots  & \vdots& \vdots& \ddots & \vdots & \vdots  \\
		0 & 0 & 0 & 0 & \cdots & 1 & 1 \\
		1 & 0 & 0 & 0 & \cdots & 0 & 1 
		\end{pmatrix}
	\end{align*}
	We have $\subrank(f) = \rank(f) = 2k+1$ by Lemma~\ref{lem:matrix-subrank-equals-rank}. On the other hand, $\symsubrank(f) < 2k+1$ by Lemma~\ref{lem:non-sym-not-max}.
\end{example}

We have discussed in Section~\ref{subsec:ind-set} how many methods for upper bounding the independence number of hypergraphs $\alpha(H)$ also upper bound the induced matching barrier $\beta(H)$. %
In the following example we see that the symmetric subrank can be strictly smaller than $\beta(H)$.

\begin{example}\label{ex:symsub-smaller-than-im}
	Let $C_5$ be the directed cycle graph on vertices $\{1, \ldots, 5\}$ as defined in Example~\ref{ex:cycles}.
	Let
	\[
		f = 
		\begin{pmatrix}
		1 & 1 & 0 & 0 & 0\\
		0 & 1 & 1 & 0 & 0\\
		0 & 0 & 1 & 1 & 0\\
		0 & 0 & 0 & 1 & 1 \\
		1 & 0 & 0 & 0 & 1 
		\end{pmatrix}
	\]
	be the adjacency matrix of $C_5$ over $\FF_2$.
	Then $\beta(C_5)$ is the size of the largest submatrix of $f$ that is an identity matrix up to permutation. We see that $\beta(C_5) = 3$. On the other hand, we compute directly that~$\symsubrank(f) =2$ over $\FF_2$.
\end{example}

\subsection{Symmetric tensors of order two (symmetric matrices)}\label{subsec:symm-complex-matrices}

We have seen in the previous section that the symmetric subrank can be strictly smaller than the subrank for non-symmetric matrices. 
For \emph{symmetric} matrices, we now prove that symmetric subrank and subrank are equal as long as the ground field is \emph{quadratically closed}\footnote{One could consider an alternative definition of symmetric subrank in which the symmetric restriction order is replaced by the following: let $f \symleq' g$ if and only if there is a matrix $A$ and diagonal matrices $D_1, \ldots, D_k, E_1, \ldots, E_k$ such that $f = (E_1\otimes \cdots \otimes E_k)(A\otimes \cdots \otimes A)(D_1\otimes \cdots \otimes D_k)g$. Under this alternative symmetric restriction preorder, the subrank and (alternative) symmetric subrank become equal for all symmetric matrices over any field. \cite{CGLM08} take a similar approach to this when dealing with the symmetric rank over the reals (which is not quadratically closed).}, meaning that every element has a square root. Algebraically closed fields are in particular quadratically closed.

\begin{theorem}\label{symmetric_matrix}
	For any symmetric matrix $f$ over a quadratically closed field $\FF\neq\FF_2$, we have $\subrank(f) = \symsubrank(f)$.
\end{theorem}

It follows from Example~\ref{example:skew_matrix} that the statement of Theorem~\ref{symmetric_matrix} indeed fails over the field~$\FF_2$ if we let $f$ be a full-rank anti-diagonal matrix.

The proof of Theorem~\ref{symmetric_matrix} relies on the following theorem.

\begin{theorem}[Ballantine \cite{BALLANTINE1968261}]\label{th:ballantine}
	Let $\FF$ be a field with size at least $3$ and $f$ be a square matrix of size $d$ over~$\FF$ that is not a nonzero skew-symmetric matrix. There is an invertible $d\times d$ matrix $B$ such that~$BfB^{T}$ is a lower triangular matrix that has exactly $\rank(f)$ nonzero elements on its diagonal.
\end{theorem}

\begin{proof}[Proof of Theorem~\ref{symmetric_matrix}]
The symmetric matrix $f$ is in particular not a nonzero skew-symmetric matrix, so we may apply Theorem~\ref{th:ballantine} to find an invertible matrix $B$ such that $BfB^T$ is lower triangular with exactly $\rank(f)$ nonzero elements on its diagonal. Since $f$ is symmetric, $BfB^T$ is also symmetric. It follows that $BfB^T$ is a diagonal matrix. Since the ground field is quadratically closed, there is a diagonal matrix $C$ such that $CBfB^TC^T$ is a diagonal matrix with only zeroes and ones on the diagonal. Then clearly $\symsubrank(f) \geq \rank(f) = \subrank(f)$, which proves the claim.
\end{proof}

\subsection{Symmetric tensors of order $k \geq 3$}

In Section~\ref{subsec:symm-complex-matrices} we proved that the symmetric subrank and subrank coincide on symmetric matrices over any quadratically closed field $\FF\neq \FF_2$ (e.g.~the complex numbers~$\CC$).
In this section we consider symmetric tensors of order $k\geq 3$. We will show that on such tensors the symmetric subrank and subrank can be different. We also give a sufficient condition for these notions to be equal (namely when the subrank is ``maximal'').

We call a tensor $f \in (\FF^d)^{\otimes k}$ \emph{symmetric} if for all $(i_1, \ldots, i_k) \in [d]^k$ and all permutations $\sigma$ of $[k]$ we have $f_{i_1, \ldots, i_k} = f_{i_{\sigma(1)}, \ldots, i_{\sigma(k)}}$. %

The following is a small example of a symmetric tensors $f$ of order three over the field $\FF_2$ for which there is strict inequality $\symsubrank(f)<\subrank(f)$. %

\begin{example}
	\label{example:tighttensor}
	Let $f = e_{1}\otimes e_2 \otimes e_3+e_{1}\otimes e_3 \otimes e_2 + e_2\otimes e_1 \otimes e_3 + e_2\otimes e_3 \otimes e_1 + e_3\otimes e_1 \otimes e_2 + e_3\otimes e_2 \otimes e_1 + e_{1}\otimes e_1 \otimes e_1$, where $e_1,e_2,e_3 \in \FF_{2}^3$ is the standard basis of $\FF_{2}^3$. It is not hard to verify that $\symsubrank(f) = 1$ while $\subrank(f) = 2$.
\end{example}

In the first preprint version of this paper, we left the construction of a symmetric tensor~$f$ over $\CC$ satisfying $\symsubrank(f) < \subrank(f)$ as an open problem. This problem is the subrank analog of Comon's conjecture about tensor rank of symmetric tensors, which was recently disproved by Shitov~\cite{Shi18}.
Subsequently we have been informed by Shitov that he can indeed construct such a tensor $f$ over $\CC$ \cite{shitov}.

Next, we prove a general sufficient condition for $\symsubrank(f) = \subrank(f)$. Namely, for symmetric complex tensors, if the subrank is maximal, then also the symmetric subrank is maximal:%

\begin{theorem}\label{thm:snap-to-max}
	 Let $f \in (\CC^d)^{\otimes k}$ be a symmetric tensor. If $\subrank(f) = d$ then $\symsubrank(f) = d$.
\end{theorem}

To prove Theorem~\ref{thm:snap-to-max} we use a simple corollary of the following theorem.

\begin{theorem}[Belitskii and Sergeichuk \cite{Belitskii07}] \label{thm:congruence_tensor}
	Let $f,f' \in (\CC^d)^{\otimes k}$ be tensors of order~$k$. If there are invertible $d\times d$ matrices $A^1,\dots,A^k$ such that $f' = (A^{\pi(1)}\otimes \dots \otimes A^{\pi(k)})f$ for all permutations $\pi \in \mathfrak{S}_k$, then there is an invertible $d\times d$ matrix $B$ such that $f' = (B\otimes \dots \otimes B)f$.
\end{theorem}

\begin{corollary} [Corollary of Theorem \ref{thm:congruence_tensor}]\label{cor:congr-tens}
	Let $f',f \in (\CC^d)^{\otimes k}$ be symmetric tensors. If there are invertible $d\times d$ matrices $A^1, \dots, A^k$ such that $f' = (A^1 \otimes \dots \otimes A^k)f$, then there is an invertible $d\times d$ matrix $B$ such that $f' = (B \otimes \dots \otimes B)f$. 
\end{corollary}
\begin{proof}[Proof of Corollary~\ref{cor:congr-tens}]
	For any permutation $\pi \in \mathfrak{S}_k$. We have
	\begin{align*}
	\sum_{j_1 \in[d],\dots,j_k \in [d]}A^{\pi(1)}_{i_1,j_1} \dots A^{\pi(k)}_{i_k,j_k}f_{j_1,\dots,j_k} &=
	\sum_{j_1 \in[d],\dots,j_k \in [d]}A^{1}_{i_{\pi^{-1}(1)},j_1} \dots A^{k}_{i_{\pi^{-1}(k)},j_k}f_{j_1,\dots,j_k} \\
	&= f'_{i_{\pi^{-1}(1)},\dots,i_{\pi^{-1}(k)}} = f'_{i_1,\dots,i_k}.
	\end{align*} 
	Therefore $f' = (A^{\pi(1)} \otimes \dots \otimes A^{\pi(k)})f$ for all $\pi \in \mathfrak{S}_k$. By using Theorem~\ref{thm:congruence_tensor}, the proof is completed.  
\end{proof}

\begin{proof}[Proof of Theorem~\ref{thm:snap-to-max}]
 Since $\subrank(f) = d$, there are $k$ matrices $A^1,\dots,A^k$ of size $d \times d$ such that $ \left \langle  d \right \rangle  = (A^{(1)} \otimes \dots \otimes A^{(k)})f$. Suppose that there is a matrix $A^{(i)}$ which is not invertible, then the rank of $i$-th flattening matrix of $(A^{(1)} \otimes \dots \otimes A^{(k)})f$ is smaller than $d-1$, that is, $\rank(\flatten_{i} ( (A^{(1)} \otimes \dots \otimes A^{(k)})f ) \leq d-1$. But the rank of all flattenings of $ \left \langle  d \right \rangle$ are equal to~$d$. Therefore all $A^{(1)},\dots,A^{(k)}$ are invertible matrices. By the above corollary, there is an invertible matrix $B$ such that $ \left \langle  d \right \rangle = (B\otimes \dots \otimes B)f$, this implies $\symsubrank(f) = d$.     
\end{proof}

\section{Asymptotic symmetric subrank}\label{sec:asymp-symm-subrank}

In Section~\ref{sec:symmetric_subrank} we introduced the symmetric subrank guided by the motivation of using this tensor parameter to upper bound the independence number of hypergraphs.
In many of these hypergraph independence problems (e.g.~the cap set problem, sunflower problem, etc.) the hypergraph under consideration has a power structure (under the strong product $\boxtimes$, which is simply the tensor product on the adjacency tensor). 
In other words, the parameter of interest in those problems is the rate of growth of the independence number of large powers of a fixed small hypergraph. This is captured by the Shannon capapcity
\[
\Theta(H) = \lim_{n\to\infty} \alpha(H^{\boxtimes n})^{1/n} = \sup_n \alpha(H^{\boxtimes n})^{1/n}.
\]
In this asymptotic context, and with upper bounding the Shannon capacity in mind, we introduce and study the asymptotic symmetric subrank.
We define the \emph{asymptotic symmetric subrank} of a tensor $f \in (\FF^d)^{\otimes k}$ as
\begin{align*}
\symasympsubrank(f) \coloneqq \limsup\limits_{n \ra \infty} \symsubrank(f^{\otimes n})^{1/n}.
\end{align*}
(The fact that we are using the lim sup rather than lim or sup is a technicality which in most relevant cases simplifies as we discuss below.)
For any tensor $f \in (\FF^d)^{\otimes k}$, since we have the basic inqualities $\symsubrank(f) \leq \subrank(f) \leq d$, we also have that $\symasympsubrank(f) \leq \asympsubrank(f)\leq d$.

Note that, because of the earlier Example~\ref{example:skew_matrix}, this lim sup cannot generally be replaced by a limit.\footnote{For the usual subrank, the \emph{asymptotic subrank} of the tensor~$f \in (\FF^d)^{\otimes k}$ was defined by Strassen as the limit
$\asympsubrank(f) = \lim_{n \ra \infty} \subrank(f^{\otimes n})^{1/n}$, which,
since $\subrank$ is super-multiplicative and $\subrank(f)\geq1$ if $f\neq 0$, equals the supremum $\sup_{n} \subrank(f^{\otimes n})^{1/n}$ (Fekete's lemma).
For the \emph{symmetric} subrank, we have to be more careful about how we define the asymptotic symmetric subrank. 
For example, in Example~\ref{example:skew_matrix} we gave a matrix~$f$ for which $f^{\otimes n}$ is symmetric if $n$ is even and skew-symmetric if $n$ is odd, and so
$\symsubrank(f^{\otimes n}) = 2^n$ if $n$ is even, and $\symsubrank(f^{\otimes n}) = 0$ when $n$ is odd.
Thus, the limit $\lim_{n\ra \infty} \symsubrank(f^{\otimes n})^{1/n}$ might not exist.}
However, we will be interested in the adjacency tensors of hypergraphs which have the special property that the coefficients on the main diagonal are all one. In that case we can replace the lim sup by a limit or supremum as follows:

\begin{proposition}
\label{limsup_sup}
	Let $f \in (\FF^d)^{\otimes k}$ be a tensor such that there is an $i \in [d]$ with $f_{i,\dots,i} = 1$. Then
	$\symasympsubrank(f) = \sup_{n} \symsubrank(f^{\otimes n})^{1/n} = \lim_{n\to\infty} \symsubrank(f^{\otimes n})^{1/n}$.\footnote{If the field $\FF$ is closed under taking $k$th roots, for instance if it is algebraically closed, then of course the theorem remains true if we replace $f_{i,\ldots, i} = 1$ by $f_{i, \ldots, i} \neq 0$.}
\end{proposition}
\begin{proof}
	Let $B \in \RR^{1\times d}$ be the $1\times d$ matrix with $B_{1,i} = 1$ and the other entries equal to~$0$. Then $ \left \langle  1 \right \rangle 
	= (B\otimes \dots \otimes B)f$. Therefore $\symsubrank(f) \geq 1$. The symmetric subrank is super-multiplicative under tensor product. Thus, by Fekete's lemma, we find the required statement that $\symasympsubrank(f) =\sup_{n} \symsubrank(f^{\otimes n})^{1/n} = \lim_{n\to\infty} \symsubrank(f^{\otimes n})^{1/n}$.
\end{proof}

The important property of $\symasympsubrank$ is that it directly gives an upper bound on the Shannon capacity of hypergraphs.
\begin{proposition}
	\label{asymptotic_bound_shannon_capacity}
	Let $H = (V,E)$ be a directed $k$-uniform hypergraph on $n$ vertices. Let $\FF$ be any field.
	Let $f \in (\FF^n)^{\otimes k}$ be a tensor such that, for every $e \in [n]^k$ if $e \not\in E$, then $f_{e_1, \ldots, e_k} = 0$, and for every $i \in [n]$, $f_{i,\dots,i} = 1$.
	Then $\Theta(H) \leq \symasympsubrank(f)$.
\end{proposition}
\begin{proof}
	By the definition of $f$, we have that $f^{\otimes n}$ satisfies the condition of Proposition~\ref{adj_boundsubrank} for the hypergraph~$H^{\boxtimes n}$. %
	Therefore $\Theta(H) = \sup_{n}(\alpha(H^{\boxtimes n}))^{1/n} \leq \sup_{n}(\symsubrank(A_{H}^{\otimes n}))^{1/n} = \symasympsubrank(A_H)$. %
\end{proof}

\subsection{Tensors of order two (matrices)}

We conjecture that the asymptotic symmetric subrank of a $k$-tensor with $k \geq 3$ can be strictly smaller than the asymptotic subrank. This cannot happen for $k = 2$. In that case we prove that there is no strict inequality, again using Theorem~\ref{th:ballantine}.

\begin{theorem}\label{thm:matrix-asymp}
	For any matrix~$f$ over a quadratically closed field $\FF \neq \FF_2$, $\asympsubrank(f) = \symasympsubrank(f)$.
\end{theorem}

\begin{proof} %
	We will use Theorem~\ref{th:ballantine}.
	We may assume that $f$ is a $d \times d$ matrix.
	Let $r = \rank(f)$. Then $\asympsubrank(f) =  \rank(f) = r$. 
	Suppose that $f$ is a skew-symmetric matrix. 
	Then we have $\symsubrank(f) = 0$ by Lemma~\ref{lem:skew-sym-zero}.
	The matrix  $f^{\otimes n}$ is symmetric if $n$ is even and skew-symmetric if $n$ is odd.
	Then by Theorem~\ref{symmetric_matrix} we have
	\begin{align*}
	\symsubrank(f^{\otimes n}) = \begin{cases}
	r^{n} \text{ if } n \text{ is even},\\
	0 \text{ otherwise}.
	\end{cases}
	\end{align*}   
	Therefore $\symasympsubrank(f) = r$. 
	Suppose that $f$ is not skew-symmetric. By Theorem~\ref{th:ballantine}, there is an invertible matrix $B$ and a lower-triangular matrix $L$ such that $BfB^{T} = L$. 
	Then $\symsubrank(f) = \symsubrank(L)$ and so $\symasympsubrank(f) = \symasympsubrank(L)$. There is a principal submatrix $A$  of $L$ of size $r$ that has exactly $r$ nonzero elements on its diagonal.
	Then $A^{\otimes n}$ is a submatrix of $L^{\otimes n}$. We choose $n=rk$ for some $k \in \NN_{\geq 1}$. Then the submatrix of $A^{\otimes n}$ with rows and columns indexed by the elements in $[r]$ of type $(n/r,\dots,n/r)$ is diagonal and has size
	\begin{align*}
	\binom{n}{n/r,\dots,n/r} \geq r^{n-\smallo(n)}.
	\end{align*}
	We conclude that $\symasympsubrank(L) \geq r$. %
\end{proof}
It follows from Example~\ref{example:skew_matrix} that the statement of Theorem~\ref{thm:matrix-asymp} is false over $\FF_2$ by taking~$f$ to be an anti-diagonal matrix with ones on the antidiagonal.

\subsection{Symmetric tensors}\label{subsec:rel-asymp}

For symmetric tensors we prove that the asymptotic symmetric subrank is equal the asymptotic subrank (as long as the field satisfies mild closedness and characteristic conditions):

\begin{theorem}\label{th:asympsubrank}
	Let $f$ be a symmetric $k$-tensor over an algebraically closed field of characteristic at least $k+1$. Then $\asympsubrank(f) = \symasympsubrank(f)$.
\end{theorem}
In particular, Theorem~\ref{th:asympsubrank} holds for any tensor over the field of complex numbers.

In fact we prove a much more general asymptotic statement about the restriction preorder~$\leq$ and the symmetric restriction preorder on symmetric tensors. We define the asymptotic restiction preorder $\asympleq$ on tensors $f,g$ by writing $f \asympleq g$ if and only if $f^{\otimes n} \leq g^{\otimes n + o(n)}$. Similarly we define the asymptotic symmetric restriction preorder $\symasympleq$ on tensors $f,g$ by writing $f \symasympleq g$ if and only if $f^{\otimes n} \symleq g^{\otimes n + o(n)}$.

\begin{theorem}\label{th:asympsymmrestr}
	For symmetric $k$-tensors $f,g$ over an algebraically closed field of characteristic at least $k+1$ we have $f \asympleq g$ if and only if $f \symasympleq g$.
\end{theorem}

It will also follow from our proof that on symmetric tensors (over an appropriate field) the asymptotic rank and symmetric asymptotic rank are equal:

\begin{theorem}\label{th:asymprank}
	Let $f$ be a symmetric $k$-tensor over an algebraically closed field of characteristic at least $k+1$.
	Then $\symrank(f) \leq 2^{k-1} \trank(f)$ and in particular
	$\asymprank(f) = \symasymprank(f)$.
\end{theorem}

For order $k = 3$ the same relation $\symrank(f) \leq 2^{k-1} \trank(f)$ (and thus $\asymprank(f) = \symasymprank(f)$) for symmetric tensors $f$ was found in~\cite{DBLP:conf/stoc/Kayal12}. %

The above three theorems are related to Comon's conjecture \cite{CGLM08}, which says that rank and symmetric rank coincide on symmetric tensors. Shitov~\cite{Shi18} gave a counterexample to Comon's conjecture. Our Theorem~\ref{th:asympsubrank}, Theorem~\ref{th:asympsymmrestr} and Theorem~\ref{th:asymprank} can be interpreted as saying that ``Comon's conjecture'' is true \emph{asymptotically}, not only for rank (Theorem~\ref{th:asymprank}), but also for subrank (Theorem~\ref{th:asympsubrank}) and the restriction preorder (Theorem~\ref{th:asympsymmrestr}). 

The proofs for all of the above will follow from three basic lemmas that we will discuss now.
 A crucial role will be played by the following $k$-tensor.
\begin{definition}[fully symmetric $k$-tensor]
	For any $k \in \NN$ let $\mathfrak{S}_k$ be the symmetric group on~$k$ elements and define the $k$-tensor~$h = \sum_{\pi \in S_k} e_{\pi(1)} \otimes \cdots \otimes e_{\pi(k)}$. We will call $h$ the \emph{fully symmetric $k$-tensor}. 
\end{definition}
For example, for $k = 3$, the tensor $h$ is given by $h = e_1 \otimes e_2 \otimes e_3 + e_1 \otimes e_3 \otimes e_2 + e_2 \otimes e_1 \otimes e_3 + e_2 \otimes e_3 \otimes e_1 + e_3 \otimes e_1 \otimes e_2 + e_3 \otimes e_2 \otimes e_1$. The tensor $h$ allows us to transform any restriction to a symmetric restriction:

\begin{lemma}\label{lem:make-sym}
	Let $f$ and $g$ be symmetric $k$-tensors over a field of characteristic at least $k+1$. If $f \geq g$, then $f \otimes h \symgeq g \otimes h$, and hence also $f \otimes h \symgeq g$, where $h$ is the fully symmetric tensor.
\end{lemma}
\begin{proof}
	Let $A_1, \ldots, A_k$ be linear maps such that $(A_1 \otimes  \cdots \otimes A_k) f = g$.
	Let $e_i^*$ denote the elements of the basis dual to the standard basis $e_i$.
	Define the linear map $B = \sum_i A_i \otimes e_i e_i^*$. Then
	\[
		(B^{\otimes k}) (f \otimes h) = k! ((A_1 \otimes \cdots \otimes A_k) f) \otimes h.
	\]
	Dividing by $k!$ proves the claim.
\end{proof}

In particular, Lemma~\ref{lem:make-sym} says that, if $f^{\otimes n} \geq \langle r\rangle$, then $f^{\otimes n} \otimes h \symgeq \langle r\rangle$ for every $n \in \NN$. Note that $h$ is a fixed tensor that is independent of $n$. Our next goal is to prove that for every~$f$ there is a constant $c \in \NN$ depending on~$f$ such that $f^{\otimes c} \symgeq h$. %
This is true in the following sense.

Recall that for any subset $S \subseteq [k]$ that is not empty and not $[k]$, any $k$-tensor $f \in V_1 \otimes \cdots \otimes V_k$ can be flattened into a 2-tensor $(\bigotimes_{i \in S} V_i) \otimes (\bigotimes_{i \in [k]\setminus S} V_i)$. For a $k$-tensor $f$ we call the ranks of these flattenings the flattening ranks of $f$. 
\begin{lemma}\label{lem:create-t}
	Let $f$ be a symmetric $k$-tensor over an algebraically closed field. Suppose that some flattening rank of $f$ is at least 2. Then there is a $c \in \NN$ such that $f^{\otimes c} \symgeq h$.
\end{lemma}

To prepare for the proof of Lemma~\ref{lem:create-t} we prove the following lemma.
 
\begin{lemma}\label{lem:remove-powers}
	Let $f$ be a symmetric $k$-tensor over an algebraically closed field. There exists a basis transformation $A \in \FF^{d \times d}$ such that  the support $S = \supp(A^{\otimes k} f) \subseteq [d]^k$ of $f$ after applying the transformation $A$ satisfies $(i, \ldots, i) \not\in S$ for every $1 \leq i \leq d-1$.
\end{lemma}
\begin{proof}
	Suppose that $f \in (\FF^d)^{\otimes k}$.
	If no element of the form $(i, \ldots, i)$ appears in $S$, then we are done. 
	Otherwise, we may assume that $(d, \ldots, d)$ appears, so that the tensor $f$ is of the form $f = f_1 e_1^{\otimes k} + f_2 e_2^{\otimes k} + \cdots + f_d e_d^{\otimes k} + {}\!$ other terms, for some coefficients $f_i$ with $f_d \neq 0$.

	We apply to $f$ the invertible linear map that maps $e_i$ to $e_i$ for $1\leq i \leq d-1$ and maps~$e_d$ to $e_d + \eps_1 e_1 + \cdots + \eps_{d-1} e_{d-1}$ for some $\eps_i \in \FF$. 
	This gives a tensor $g \in (\FF^d)^{\otimes k}$ that is isomorphic to~$f$ and of the form $g = (f_1 + \eps_1^k f_d) e_1^{\otimes k} + \cdots + (f_{d-1} + \eps_{d-1}^k f_d)e_{d-1}^{\otimes k} + {}\!$ other terms. 
	Since~$f_d$ is nonzero and the ground field is algebraically closed, there are values for the $\eps_i$ such that $f_i + \eps_i^k f_i$ is zero for every $1\leq i \leq d-1$, in which case $(i, \ldots, i)$ does not appear in the support of $g$ for every $1 \leq i \leq d-1$.
\end{proof}

\begin{proof}[Proof of Lemma~\ref{lem:create-t}]
	Let $f \in (\FF^d)^{\otimes k}$.
    By Lemma~\ref{lem:remove-powers} we may assume that $(i,\ldots, i)$ does not appear in the support~$S = \supp(f) \subseteq [d]^k$ of $f$ for $1 \leq i \leq d-1$. 
	For every element $s \in S$ we define its type $(y_1, \ldots, y_d)$ by letting~$y_i$ be the number of times that $i$ appears in $s$.
	Let $Y$ be the set of types of elements of $S$.
	Since some flattening rank of $f$ is at least 2, we cannot have that $S = \{(d, \ldots, d)\}$. 
	Thus without loss of generality there is a type $y \in Y$ %
	that satisfies $1 \leq y_1 \leq k-1$ and such that for every other type $y' \in Y$ it holds that $y_1' \leq y_1$ (maximality assumption).
    Let $R \subseteq [d]^k$ be the set of all $k$-tuples in $[d]^k$ of type $y$. %
    Let $A$ be the $|R| \times k$ matrix with rows given by the elements of $R$, in some arbitrary order. 
    Let $C$ be the set of columns of~$A$. 
	Note that in any $s \in S$ the element $1$ can appear at most $y_1$ times by our maximality assumption.

    We claim that $f^{\otimes |R|}$ restricts symmetrically to the fully symmetric $k$-tensor $h$ by zeroing out all basis elements that are not in $C$.
    To prove this we need to show that for any choice of $k$ elements $v_1, \ldots, v_k$ in~$C$, if for every $i$ we have that $((v_1)_i, \ldots, (v_k)_i) \in S$, then $v_1, \ldots, v_k$ are all different.

    By construction of $C$, for any $y_1$ distinct elements $v_1, \ldots,v_{y_1}$ of $C$ there is an $1 \leq i \leq |R|$ such that $(v_1)_i = \cdots = (v_{y_1})_i = 1$. Thus also for any $y_1$ (not necessarily distinct) elements $v_1, \ldots, v_{y_1}$ of $C$ there is an $1 \leq i \leq |R|$ such that $(v_1)_i = \cdots = (v_{y_1})_i = 1$.

    Let $v_1, \ldots, v_k$ be an arbitrary collection of elements of $C$. Suppose that $v_1 = v_2$. By the previous argument we know that there is an $1 \leq i \leq |R|$ such that $(v_2)_i = \cdots = (v_{y_1+1})_i = 1$. From the assumption $v_1 = v_2$ it follows that $(v_1)_i = (v_2)_i = \cdots = (v_{y_1+1})_i = 1$. However, we picked the type $(y_1, \ldots, y_d)$ such that $y_1$ is maximal and $y_1 \leq k-1$. The element $1$ appears at least~$y_1 + 1$ times in $((v_1)_i, \ldots, (v_k)_i)$. Therefore $((v_1)_i, \ldots, (v_k)_i)$ is not in $S$.
\end{proof}

\begin{proof}[Proof of Theorem~\ref{th:asympsymmrestr}]
	Suppose that $f \asympgeq g$. This means that $f^{\otimes m + o(m)} \geq g^{\otimes m}$. We know from Lemma~\ref{lem:create-t} that there is a constant~$c \in \NN$, depending only on $f$, such that $f^{\otimes c} \symgeq h$. By Lemma~\ref{lem:make-sym} we then have
	\[
		f^{\otimes m + o(m)} \otimes f^{\otimes c} \symgeq f^{\otimes m + o(m)} \otimes h \symgeq g^{\otimes m}.
	\]
	This means $f \symasympgeq g$, which proves the claim.
\end{proof}

Although essentially Theorem~\ref{th:asympsubrank} and Theorem~\ref{th:asymprank} can be proven abstractly from  Theorem~\ref{th:asympsymmrestr}, we will give the (simple) proofs separately in terms of the above lemmas for the convenience of the reader and to get the precise statement of Theorem~\ref{th:asymprank}:

\begin{proof}[Proof of Theorem~\ref{th:asympsubrank}]
Suppose that $\subrank(f^{\otimes n}) \geq r$. Then $f^{\otimes n} \geq \langle r\rangle$. By Lemma~\ref{lem:create-t} there is a constant~$c \in \NN$, depending only on $f$, such that $f^{\otimes c} \symgeq h$. By Lemma~\ref{lem:make-sym} we then have
\[
	f^{\otimes n + c} \symgeq f^{\otimes n} \otimes f^{\otimes c} \symgeq f^{\otimes n} \otimes h \symgeq \langle r\rangle.
\]
Thus $\symsubrank(f^{\otimes n + c}) \geq r$, which implies the claim.
\end{proof}

\begin{proof}[Proof of Theorem~\ref{th:asymprank}]
Suppose that $\trank(f) \leq r$. Then $f \leq \langle r \rangle$. 
Let $s = \symrank(h)$ be the symmetric rank of the fully symmetric tensor $h$ and note that $s$ is a constant depending only on $k$, the order of $f$. In fact, $s \leq 2^{k-1}$, which follows from the known identity
\[
	h = \frac{1}{2^{k-1}}\sum_{\varepsilon_i = \pm1} \Bigl(\prod_{i=2}^k \varepsilon_i\Bigr) (e_1 + \varepsilon_2 e_2 + \varepsilon_3 e_3 + \cdots + \varepsilon_k e_k)^{\otimes k}
\]
in which the sum goes over $\eps_2, \ldots, \eps_k = \pm 1$.
We refer to \cite[Lemma B.2.3]{goodman2009symmetry} for a proof of this identity. See also \cite[Proposition~11.6]{DBLP:journals/focm/LandsbergT10}.
Then
\[
\langle r s \rangle = \langle r\rangle \otimes \langle s \rangle \symgeq \langle r \rangle \otimes h \symgeq f.
\]
Thus $\symrank(f) \leq r s$, which implies the first claim. Then, since $s$ is constant, it follows that $\trank(f^{\otimes n}) \leq \symrank(f^{\otimes n}) \leq \trank(f^{\otimes n}) s$ for every $n \in \NN$, which implies the second claim.
\end{proof}

\section{Asymptotic spectrum of symmetric tensors}\label{sec:asympspec}

In Section~\ref{sec:symmetric_subrank} we introduced the symmetric subrank and in Section~\ref{sec:asymp-symm-subrank} we introduced the asymptotic symmetric subrank, both motivated by the problem of upper bounding the independence number of hypergraphs (with the asymptotic symmetric subrank in particular being relevant for capacity-type questions, where the hypergraphs at hand have a power structure). We proved several equalities and separations for these parameters.

In this section we continue our analysis of the asymptotic symmetric subrank in a general fashion that also allows us to discuss the asymptotic symmetric rank and the asymptotic symmetric restriction preorder (which we will define).

At the core of this section is the duality theory of Strassen introduced and studied in \cite{Str86,Str88,Str88,Str91,tobler,burg} (see also \cite{CVZ18} and~\cite{Zui18_thesis}) that gives a \emph{dual formulation} for the (non-symmetric) asymptotic subrank, asymptotic rank and asymptotic restriction preorder in terms of the \emph{asymptotic spectrum of tensors}.
The asymptotic subrank of $f \in \FF^{n_1} \otimes \cdots \otimes \FF^{n_k}$ is defined as $\asympsubrank(f) = \lim_{n\to\infty} \subrank(f^{\otimes n})^{1/n}$, the asymptotic rank is defined as $\asymprank(f) = \lim_{n\to\infty} \rank(f^{\otimes n})^{1/n}$ and the asymptotic restriction preorder is defined by $f \asympleq g$ if and only if $f^{\otimes n} \leq g^{\otimes (n + o(n))}$.
As an application of the results of Section~\ref{subsec:rel-asymp} we prove a strong connection between this theory and the natural symmetric variation.

The asymptotic spectrum of tensors (for any fixed $k \in \NN$ and field $\FF$) is defined as the set $X$ of all real-valued maps from $k$-tensors over $\FF$ to the nonnegative reals that are additive under the direct sum, multiplicative under the tensor product, monotone under the restriction preorder and normalized to $1$ on the diagonal tensor $\langle 1\rangle$ of size one. The duality theory says that: the asymptotic rank equals the pointwise maximum over all elements in the asymptotic spectrum of tensors, the asymptotic subrank equals the pointwise minimum over all elements in the asymptotic spectrum of tensors, and the asymptotic restriction preorder is characterized by $f \asympleq g$ if and only if for every $\phi$ in the asymptotic spectrum $X$ it holds that $\phi(f) \leq \phi(g)$.

\subsection{Asymptotic spectrum duality}

We introduce the asymptotic spectrum of \emph{symmetric} tensors as the natural symmetric variation on Strassen's asymptotic spectrum of tensors, to give a duality theory for the asymptotic symmetric (sub)rank and restriction preorder. We have defined the asymptotic symmetric subrank before. The asymptotic symmetric rank is similarly defined as $\symasymprank(f) = \lim_{n\to\infty} \symrank(f^{\otimes n})^{1/n}$ and the asymptotic symmetric restriction preorder is defined by $f \symasympleq g$ if and only if $f^{\otimes n} \symleq g^{\otimes (n + o(n))}$.

We define the asymptotic spectrum of symmetric tensors (for any fixed $k \in \NN$ and field $\FF$) as the set $X_\mathrm{s}$ of all real-valued maps from symmetric $k$-tensors over $\FF$ to the nonnegative reals that are additive under the direct sum, multiplicative under the tensor product, monotone under the \emph{symmetric} restriction preorder, and normalized to $1$ on the diagonal tensor $\langle 1\rangle$. It follows readily from the general part of the theory in \cite{Str88} (see also \cite{Zui18_thesis}) that the asymptotic spectrum of symmetric tensors $X_\mathrm{s}$ gives a dual formulation for the asymptotic symmetric subrank, asymptotic symmetric rank and asymptotic symmetric restriction preorder.
\begin{theorem} \label{thm:asymptoticspectrum}
	Let $\FF$ be an algbraically closed field of characteristic at least $k+1$. Let $X_\mathrm{s}$ be the asymptotic spectrum of symmetric $k$-tensors.
	Let $f$ and $g$ be symmetric $k$-tensors. Then
	\begin{align*}
	\symasympsubrank(f) &= \min_{\phi \in X_\mathrm{s}}\phi(f),\\
		\symasymprank(f) &= \max_{\phi \in X_\mathrm{s}}\phi(f),\\
		f \symasympleq g &\iff \forall \phi \in X_\mathrm{s},\, \phi(f) \leq \phi(g).
	\end{align*} 
\end{theorem}

We will not give the proof of Theorem~\ref{thm:asymptoticspectrum} as it follows along the same lines as the original proof in \cite{Str88} (see also \cite{Zui18_thesis}).
The bulk of the proof is to show that the symmetric restriction preorder is a so-called ``good preorder'' (\cite{Str88}) or Strassen preorder (\cite{Zui18_thesis}).
The only non-standard ingredient for the proof is the fact that for every nonzero symmetric $k$-tensor $f$ either $f$ is equivalent to $\langle 1\rangle$ or $\symasympsubrank(f) > 1$, which follows from Theorem~\ref{th:asympsubrank} and the fact that this property holds for $\asympsubrank$.

\subsection{Surjective restriction from the asymptotic spectrum}

The results of Section~\ref{subsec:rel-asymp} answer a structural question: how are the asymptotic spectrum of tensors $X$ and the asymptotic spectrum of symmetric tensors $X_\mathrm{s}$ related? 
One relation is clear: for every element $\phi \in X$ the restriction of $\phi$ to symmetric tensors is an element of $X_\mathrm{s}$. We thus have the restriction map $r : X \to X_\mathrm{s}$ that maps $\phi \in X$ to the restriction of $\phi$ to symmetric tensors. We prove:

\begin{theorem}\label{th:restriction}
	The restriction map $r : X \to X_\mathrm{s}$ is surjective.
\end{theorem}

Theorem~\ref{th:restriction} has two readings: (1) if we understand what the elements are of the asymptotic spectrum of tensors $X$, then we also understand what the elements are of the asymptotic spectrum of symmetric tensors $X_\mathrm{s}$ by restriction, and (2) for any element $\psi \in X_\mathrm{s}$ there is an \emph{extension} $\phi \in X$ such that $\phi$ restricts to $\psi$.

Theorem~\ref{th:restriction} follows from our Theorem~\ref{th:asympsymmrestr} together with an application of the following powerful theorem from the theory of asymptotic spectra. The theorem uses the notion of a good preorder or Strassen preorder for which we refer the reader to the literature.

\begin{theorem}[{\cite{Str88}, \cite[Corollary 2.18]{Zui18_thesis}}]\label{th:surj}
	Let $S$ be a semiring with a Strassen preorder~$P$. Let $T$ be a subsemiring of $S$. Then the restriction map from the asymptotic spectrum of $S$ to the asymptotic spectrum of $T$ is surjective.
\end{theorem}

\begin{proof}[Proof of Theorem~\ref{th:restriction}]
	We give a sketch of the proof. The proof is an application of Theorem~\ref{th:surj}. Let $S$ be the semiring of $k$-tensors and let $P$ be the asymptotic restriction preorder. This is a Strassen preorder. Let $T$ be the subsemiring of $S$ of symmetric $k$-tensors. Then Theorem~\ref{th:surj} implies that the restriction map from the asymptotic spectrum of $S$ with the asymptotic restriction preorder to the asymptotic spectrum of $T$ with the asymptotic restriction preorder is surjective. Since the asymptotic restriction preorder on symmetric tensors coincides with the asymptotic symmetric restriction preorder by Theorem~\ref{th:restriction}, the claim follows.
\end{proof}

To conclude and summarize, the asymptotic spectrum of tensors $X$ and the asymptotic spectrum of symmetric tensors $X_\mathrm{s}$ are tightly related since the restriction map from the first to the second is surjective. What are the elements of $X$ and $X_\mathrm{s}$? A long line of work \cite{Str86,Str88,Str88,Str91,Str05,tobler,burg, CVZ18, christandl2020weighted} has been devoted to this question. Our best understanding is for the case that the ground field $\FF$ is the complex numbers\footnote{It is known that the asymptotic spectrum can only depend on the characteristic of the field \cite{Str88}.} and that is what we will focus our discussion on here and in the next section.

The known elements in $X$ (over the complex numbers) are a family of functions called the \emph{quantum functionals}. These were introduced in \cite{CVZ18} and are based on an information-theoretic and representation-theoretic study of powers of tensors. The quantum functionals more precisely form a continuous family $F^\theta$ indexed by probability distributions $\theta$ on $[k]$. This family includes the flattening ranks, but also includes more interesting functions that are properly real-valued which reveal asymptotic information that the flattening ranks do not reveal. It is possible but not known whether the quantum functionals are all elements of~$X$. Proving this is a central open problem of the theory. In particular, the quantum functionals being all elements of $X$ would imply that the matrix multiplication exponent $\omega$ equals 2, which would be a breakthrough result in complexity theory.

We may restrict the quantum functionals to symmetric tensors to find an infinite family of elements in $X_\mathrm{s}$\footnote{Restricting the quantum functionals to symmetric tensors may make some of the functionals coincide, but known results imply that the resulting family is still a continuous family. Namely, there is symmetric tensor $T$, called the W-tensor, for which it is known that $\{F(T) : F \in X\}$ equals the closed interval $[3/2^{2/3}, 2]$ and thus also $\{F(T) : F \in X_\mathrm{s}\} = [3/2^{2/3}, 2]$.}. Since we do not know whether the quantum functionals are all elements of~$X$, we can, however, not conclude from Theorem~\ref{th:restriction} that their restriction gives all elements of~$X_\mathrm{s}$. 

What we will do in the next section is give a natural construction of a single element in $X_\mathrm{s}$ following the same ideas as for the construction of the quantum functionals but applied directly to the symmetric restriction preorder. This single element we call the \emph{symmetric quantum functional}. What we then find is that this symmetric quantum functional on symmetric tensors in fact equals the \emph{uniform} quantum functional $F^{(1/k, \ldots, 1/k)}$. Thus we do not find a new element in~$X_\mathrm{s}$, but we do find a different description of the uniform quantum functional restricted to symmetric tensors, and this might be algorithmically beneficial. This symmetric quantum functional is the \emph{pointwise smallest} element among all elements in $X_\mathrm{s}$ that we currently know, and from previous work it follows that it equals the asymptotic slice rank (on symmetric tensors). Having discussed the plan we will now go into the details in the next section.

\section{Symmetric quantum functional}\label{sec:symmetricquantum}

In Section~\ref{sec:asympspec} we introduced the asymptotic spectrum of symmetric tensors $X_\mathrm{s}$ and proved a duality theorem for the asymptotic symmetric (sub)rank and restriction preorder in terms of~it. 

We use the ideas of the construction of the quantum functionals $F^{\theta} \in X$ from \cite{CVZ18} to construct the \emph{symmetric quantum functional} $F \in X_\mathrm{s}$ over the field of complex numbers. Let us from now on fix the base field to be the field of complex numbers. In fact we will take a more general approach and define the symmetric quantum functional not just for symmetric tensors but for arbitrary tensors.

Before recalling the definition of the quantum functionals $F^\theta$ and giving the new definition of the symmetric quantum functional $F$, here is what we will find. For symmetric tensors we will show that:

\begin{theorem}\label{th:main1}
	On symmetric tensors $F = F^{(1/k, \ldots, 1/k)}$.
\end{theorem}

This gives an alternative description of the uniform quantum functional $F^{(1/k, \ldots, 1/k)}$, which may have algorithmic benefits. 

In particular on symmetric tensors the symmetric quantum functional is in the asymptotic spectrum of symmetric tensors $X_\mathrm{s}$.

\begin{theorem}\label{th:main2}
	On symmetric tensors we have $F \in X_\mathrm{s}$.
\end{theorem}

For general tensors we find the following.

\begin{theorem}\label{th:main3}
	On arbitrary tensors we have $F \geq F^{(1/k, \ldots, 1/k)}$.	
\end{theorem}

In particular, since $F^{(1/k, \ldots, 1/k)} \geq \asympsubrank$ (because every quantum functional $F^\theta$ is in the asymptotic spectrum of tensors~$X$), we also find $F \geq \asympsubrank$ on arbitrary tensors.
However, via a known connection from \cite{CVZ18} between the quantum functionals and the asymptotic slice rank (the pointwise minimum $\min_\theta F^\theta$ equals the asymptotic slice rank), we find that~$F$, as a method to upper bound the Shannon capacity of hypergraphs, suffers from the induced matching barrier. %

\subsection{From quantum functionals to symmetric quantum functional}

Before defining the quantum functionals and symmetric quantum functional and giving the proofs of the above, we must introduce some standard notation.
Let $\mathcal{H}$ be a complex finite-dimensional Hilbert space with dimension $\dim(\cH)=d$. Thus $\cH \cong \CC^d$. A \emph{state} or \emph{density operator} on $\cH$ is a positive semidefinite linear map $\rho: \cH \ra \cH$ with $\tr(\rho) = 1$. Let $S(\cH)$ be the set of states on $\cH$. 
For $\rho \in S(\cH)$, let $\spec(\rho) = (\lambda_1,\dots,\lambda_d)$ be the sequence of eigenvalues of~$\rho$, ordered non-increasingly, that is, $\lambda_1 \geq \dots \geq \lambda_d$. 
Since $\tr(\rho) = 1$, the sequence of eigenvalue of~$\rho$ is a probability distribution. It thus makes sense to define $H(\spec(\rho))\coloneqq -\sum_{j=1}^{d}\lambda_j\log \lambda_{j}$. The \emph{von Neumann entropy} of~$\rho$ is defined as $H(\rho) = -\tr(\rho \log \rho) = H(\spec(\rho))$. 

Given a state $\rho$ on $\cH_1 \otimes \dots \otimes \cH_k$, the $j$th \emph{marginal} is the element $\rho_j =\tr_{\cH_1\dots \cH_{j-1}\cH_{j+1}\dots \cH_k}(\rho)$ obtained from $\rho$ by a partial trace. The $j$th marginal is itself a state, that is, $\rho_j \in \cS(\cH_j)$. Consider a nonzero element $f \in \cH^{\otimes k}$. Then $\rho(f) = \frac{ff^{\dagger}}{\|f\|^2} \in \cS(\cH^{\otimes k})$, where $f^{\dagger}$ denotes the conjugate transpose of $f$, and we can consider the $j$th marginal $\rho_{j}(f) \in \cS(\cH)$. Let $\GL(d)$ denote the set of invertible matrices acting on $\cH$. For a tensor $f \in \cH^{\otimes k}$, let $\overline{\GL(d) \cdot f}$ be the Euclidean closure (or equivalently Zariski closure) of the orbit $\{(g \otimes \dots \otimes g) f: g \in \GL(d) \}$.

We begin with the definition of the symmetric quantum functional. 

\begin{definition}[Symmetric quantum functional]
	Let $f \in \cH^{\otimes k}$ be nonzero. We define the \emph{symmetric quantum functional} $F$ by $F(f) = 2^{E(f)}$ where
	\begin{align*}
		E(f) &= \max \{ H(p) : p \in \Delta(f) \},
	\end{align*}
	where
	we define the subset $\Delta(f)\subseteq \RR^{d}$, for $d = \dim(\cH)$, as
	\begin{align*}
	\Delta(f) = \Bigl\{\spec \Bigl( \frac{\rho_1(s) + \dots+ \rho_{k}(s)}{k}\Bigr): s \in \overline{\GL(d) \cdot f} \setminus \{0 \} \Bigr\}.
	\end{align*}
\end{definition}

From the work of \cite{Ness} and~\cite{Bri87} it follows that $\Delta(f)$ is a convex polytope.

The definition of the symmetric quantum functional $F$ is inspired by the family of quantum functionals $F^{\theta}$. %
Our main results about the symmetric quantum functional give precise relations between~$F$ and $F^{\theta}$.

\begin{definition}[Quantum functionals]\label{quantum_functional_def}
	Let $\theta \in \cP([k])$ and let $f \in \cH^{\otimes k}$. The quantum functionals are defined by $F^{\theta}(f)= 2^{E^{\theta}(f)}$ where
	\begin{align*}
	E^{\theta}(f)& =  \max \left\{ \sum_{i=1}^{s}\theta(i)H(\rho_{i}(s)): s \in \overline{ \GL(d)^{\times k} \cdot f} \setminus \{0 \} \right\}
	\end{align*}        
	where $\GL(d)^{\times k} \cdot f = \{(g_1\otimes \dots \otimes g_k) \cdot f: g_1,\dots,g_k \in \GL(d) \}$.
\end{definition}

There is an asymptotic connection between the quantum functionals and the slice rank, which we will be using.
\begin{theorem}[\cite{CVZ18}] \label{prop: asymptotic_slice_rank_quantum}
	For any $f\in \cH^{\otimes k}$ the limit $\lim_{n \ra \infty} \slicerank(f^{\otimes n})^{1/n}$ exists and equals the minimization $\min_{\theta \in \cP([k])}F^{\theta}(f)$. 	
\end{theorem}

\subsection{Properties and relations}

Now we are ready to state the precise results on the symmetric quantum functional. These results in particular imply the three main results that we stated above in Theorem~\ref{th:main1}, Theorem~\ref{th:main2} and Theorem~\ref{th:main3}.

First of all, we prove that the symmetric quantum functional is at least the uniform quantum functional, and we show that the latter can be obtained as the regularization of the former:

\begin{theorem} \label{thm: bound_quantum_functional}
	Let $f \in \cH^{\otimes k}$ be any tensor. Let $\theta = \left( \frac{1}{k},\dots, \frac{1}{k} \right)$. Then
	\begin{align*}
		\lim_{n\ra \infty} \slicerank(f^{\otimes n})^{1/n} \leq F^{\theta}(f) \leq F(f) \text{ and } \lim_{n\ra \infty} F(f^{\otimes n})^{1/n} = \inf_{n}F(f^{\otimes n})^{1/n} = F^{\theta}(f). 
	\end{align*}
\end{theorem}

Second, on symmetric tensors we prove the following even stronger connection between the symmetric quantum functional and the uniform quantum functional:

\begin{theorem}\label{thm:sym-unifom}
	Let $f \in \cH^{\otimes k}$ be a symmetric tensor. %
	Then
\begin{align*}
\lim_{n\ra \infty} \slicerank(f^{\otimes n})^{1/n} = F^{(1/k, \ldots, 1/k)}(f) = F(f). 
\end{align*}
\end{theorem}

Third, from the equality $F = F^{(1/k, \ldots, 1/k)}$ on symmetric tensors (Theorem~\ref{thm:sym-unifom}), and the known properties of $F^{(1/k, \ldots, 1/k)}$, we directly obtain all of the following properties of the symmetric quantum functional $F$:

\begin{corollary}\label{cor:asympspec}
	For any symmetric~$f \in (\CC^{d})^{\otimes k}$ and $g \in (\CC^{e})^{\otimes k}$, and any $r \in \NN$, we have
	\begin{enumerate}
		\item $F(\left \langle r \right \rangle) =r$
		\item $F(f \oplus g) = F(f)+F(g)$
		\item $F(f\otimes g) = F(f)F(g)$
		\item if $f\symleq g$ then $F(f) \leq F(g)$.
	\end{enumerate}
\end{corollary}

Therefore, the symmetric quantum functional belongs to the asymptotic spectrum of symmetric tensors $X_\mathrm{s}$, which we discussed in Section~\ref{sec:asympspec}.\par
\bigskip
We will now give the proofs of the above Theorem~\ref{thm: bound_quantum_functional} and Theorem~\ref{thm:sym-unifom}. 
We will need another
characterization of $\Delta(f)$ from representation theory. Let $\lambda$ be a partition of $nk$ into at most $d$ parts. We denote this by $\lambda \vdash_{d} nk$. Then $\bar{\lambda}\coloneqq\lambda/nk = (\lambda_1/nk,\dots,\lambda_d/nk)$ is a probability distribution on $[d]$. The symmetric group $\mathfrak{S}_{nk}$ acts on $(\cH^{\otimes k})^{\otimes n}$ by permuting the tensor legs, that is, $\pi \cdot (v_1 \otimes \dots \otimes v_{nk}) = v_{\pi^{-1}(1)} \otimes \dots \otimes v_{\pi^{-1}(nk)}$ for $\pi \in \mathfrak{S}_{nk}$. The general linear group $\GL(d)$ acts on $(\cH^{\otimes k})^{\otimes n}$ via the diagonal embedding $\GL(d) \ra \GL(d)^{\times nk}: g \mapsto (g,\dots,g)$, that is, $g\cdot v = (g\otimes \dots \otimes g)v$ for $g\in \GL(d), v \in (\cH^{\otimes k})^{\otimes n}$. The Schur--Weyl duality gives a decomposition of the space $(\cH^{\otimes k})^{\otimes n}$ into direct sum of irreducible $\mathfrak{S}_{nk} \times \GL(d)$ representations. More precisely,
\begin{align*}
(\cH^{\otimes k})^{\otimes n} \cong \bigoplus_{\lambda \vdash_{d} nk} [\lambda] \otimes \mathbb{S}_{\lambda}(\cH), 
\end{align*} 
where $\mathbb{S}_{\lambda}(\cH)$ is an irreducible representation of $\GL(d)$ and $[\lambda]$ is an irreducible representation of $\mathfrak{S}_{nk}$. Let $P_{\lambda}: (\cH^{\otimes k})^{\otimes n} \ra (\cH^{\otimes k})^{\otimes n}$ be the equivariant projector onto the isotypical component of type $\lambda$, that is, onto the subspace of $(\cH^{\otimes k})^{\otimes n}$ which isomorphic to $\mathbb{S}_{\lambda}(\cH) \otimes [\lambda]$. Based on~\cite{Bri87},~\cite{Franz},~\cite{Str05} (and also \cite[Section 2.1]{MichaelThesis} and \cite[Chapter 6]{Zui18_thesis}) we have the following characterization of $\Delta(f)$.
\begin{lemma}
	\label{moment_polytope_dual} 
	The polytope $\Delta(f)$ is the Euclidean closure of the set
	\begin{align*}
	\left\{\frac{\lambda}{nk}: \exists n \in \NN_{\geq 1}, \lambda \vdash_{d}nk, P_{\lambda}f^{\otimes n} \neq 0 \right\}.
	\end{align*}
\end{lemma}
\begin{proof}
	See Appendix~\ref{moment polytope}.
\end{proof}

\begin{proof}[Proof of Theorem~\ref{thm: bound_quantum_functional}]
We decompose $\cH^{\otimes n}$ into a direct sum of irreducible $\mathfrak{S}_n \times \GL(d)$ representations as
\begin{align} \label{Schur--Weyl_duality_decomposition}
\cH^{\otimes n} \cong \bigoplus_{\lambda \vdash_{d} n}[\lambda] \otimes \mathbb{S}_{\lambda}(\cH).
\end{align} 
Let $P_{\lambda}$ be the equivariant projector onto the isotypical component of type $\lambda$. 
The uniform quantum functional $F^{(\frac{1}{k},\dots,\frac{1}{k})}(f)$ has another characterization as follows \cite{CVZ18}:
\begin{align*}
F^{(\frac{1}{k},\dots,\frac{1}{k})}(f) = \sup \Bigl\{\Bigl(\prod_{i=1}^{k}\dim[\lambda^{i}] \Bigr)^{1/kn} : \exists n \in \NN_{\geq 1},\, \lambda^{i} \vdash_{d}n, (P_{\lambda^1}\otimes \dots \otimes P_{\lambda^k})f^{\otimes n} \neq 0 \Bigr\}.
\end{align*}
For the symmetric quantum functional, using the characterization of $\Delta(f)$ from representation theory, we have 
\begin{align*}
F(f) = \sup \left\{\left(\dim[\lambda] \right)^{1/kn}: \exists n \in \NN_{\geq 1}, \lambda \vdash kn, P_{\lambda}f^{\otimes n} \neq 0 \right\}. 
\end{align*}
We may write $(\cH^{\otimes n})^{\otimes k}$ as a direct sum of irreducibles under the action of $\mathfrak{S}_{nk}$ as
\begin{align} 
\label{decomposition_symmetric}
\left(\cH^{\otimes n} \right)^{\otimes k} \cong \bigoplus_{\lambda \vdash_{d} kn} \left([\lambda] \right)^{\oplus m_\lambda}
\end{align}
where $m_\lambda = \dim \left(\mathbb{S}_{\lambda}(\cH) \right)$. 
We view $\mathfrak{S}_n^{\times k}$ naturally as a subgroup of $\mathfrak{S}_{nk}$.
For any $\lambda \vdash_{d} kn$ the restriction of $[\lambda]$ to the action of $\mathfrak{S}_n^{\times k}$ decomposes further as a direct sum of irreducibles under the action of $\mathfrak{S}_n^{\times k}$, so that
\begin{align}
\label{decompose_lambda}
[\lambda] \cong \bigoplus_{\lambda^1 \vdash_{d} n, \dots, \lambda^k \vdash_{d} n} \left([\lambda^1] \otimes \dots \otimes [\lambda^k] \right)^{\oplus c_{\lambda^1,\dots,\lambda^{k}}}
\end{align}
where $c_{\lambda^1,\dots,\lambda^{k}}$ are multiplicities.
Let $\lambda^1 \vdash_{d} n, \dots, \lambda^k \vdash_{d} n$. Then $[\lambda^1] \otimes \dots \otimes [\lambda^k]$ is irreducible representation of $\mathfrak{S}_n \times \dots \times \mathfrak{S}_n$. This gives us the finer decomposition into irreducibles under the action of $\mathfrak{S}_n^{\times k}$ as
\begin{align}
\left(\cH^{\otimes n} \right)^{\otimes k} \cong \bigoplus_{\lambda^1 \vdash_{d}n,\dots, \lambda^k \vdash_{d}n} \left([\lambda^1] \otimes \dots \otimes [\lambda^k] \right)^{\oplus m_{\lambda^1,\dots,\lambda^k}}
\end{align}
where $m_{\lambda^1,\dots,\lambda^k} = \prod_{i=1}^{k}\dim \left(\mathbb{S}_{\lambda_i}(\cH) \right)$. 

For any $n$ and $\lambda^{1} \vdash_{d}n ,\dots,\lambda^k \vdash_{d} n$ such that $\left( P_{\lambda^1} \otimes \dots \otimes P_{\lambda^k}\right)f^{\otimes n} \neq 0$ the equivariant projection of $f^{\otimes n}$ on 
\[
	\left([\lambda^1] \otimes \dots \otimes [\lambda^k] \right)^{\oplus m_{\lambda^1,\dots,\lambda^k}}
\]
is non-zero. 
From \eqref{decompose_lambda} we know that there is a $\lambda \vdash_{d} kn$ such that $[\lambda^1] \otimes \dots \otimes [\lambda^k]$ is a subspace of $[\lambda]$. For this $\lambda$ it holds that $P_{\lambda}f^{\otimes n} \neq 0$ and $\dim[\lambda] \geq \prod_{i=1}^{k} \dim\left( [\lambda^i]\right)$. This implies $F(f) \geq F^{(\frac{1}{k},\dots,\frac{1}{k})}(f)$.

For any tensor $s \in \cH^{\otimes k}$, it follows from a standard property of the von Neumann entropy~\cite[Theorem~11.10]{NC02} that
\begin{align*}
H\biggl(\frac{\rho_1(s) + \dots + \rho_{k}(s)}{k}\biggr) \leq \frac{H(\rho_{1}(s))+\dots + H(\rho_{k}(s))}{k} + \log k.
\end{align*}
This implies $F(f) \leq kF^{(\frac{1}{k},\dots,\frac{1}{k})}(f)$. Thus we have proven that 
\[
	F^{(\frac{1}{k},\dots,\frac{1}{k})}(f) \leq F(f) \leq k F^{(\frac{1}{k},\dots,\frac{1}{k})}(f)
\]
holds for every tensor $f$. In particular, applying this to the tensor power $f^{\otimes n}$ we have
\begin{align*}
F^{(\frac{1}{k},\dots,\frac{1}{k})}(f^{\otimes n}) \leq F(f^{\otimes n}) \leq kF^{(\frac{1}{k},\dots,\frac{1}{k})}(f^{\otimes n}).
\end{align*}
Since $F^{(\frac{1}{k},\dots,\frac{1}{k})}$ is multiplicative \cite{CVZ18}, we have
\begin{align*}
F^{(\frac{1}{k},\dots,\frac{1}{k})}(f) \leq F(f^{\otimes n})^{1/n} \leq k^{1/n}F^{(\frac{1}{k},\dots,\frac{1}{k})}(f).
\end{align*}
Taking $n \ra \infty$, we obtain $\lim_{n\ra \infty}F(f^{\otimes n})^{1/n} = F^{(\frac{1}{k},\dots,\frac{1}{k})}(f)$.

Finally, since $F$ is sub-multiplicative (Appendix~\ref{proof_quantumfunctional}), the limit $\lim_{n\ra \infty}F(f^{\otimes n})^{1/n}$ equals the infimum $\inf_{n} F(f^{\otimes n})^{1/n}$ by Fekete's lemma.
\end{proof}

\begin{proof}[Proof of Theorem~\ref{thm:sym-unifom}]
	Let $S$ be the set of symmetric tensors in $\overline{\left( \GL(d)^{\times k}\right) \cdot f} \setminus \{0\} $. Since~$f$ is a symmetric tensor, for any matrix $A$ the tensor $(A \otimes \dots \otimes A)f$ is also a symmetric tensor. Therefore $\overline{\GL(d) \cdot f} \setminus \{0\} \subseteq S$. Moreover, if $s$ is a symmetric tensor then all marginal density matrices are equal: $\rho_1(s) = \dots = \rho_{k}(s)$. Thus, for any $\theta \in \cP([k])$, we have $E^{\theta}(s) = \rho_1(s)$. This implies $F(f) \leq F^{\theta}(f)$ since both $F(f)$ and $F^{\theta}(f)$ are given by the supremum of the same function and for $F(f)$ the supremum is taken over a smaller set than for $F^{\theta}(f)$. 
	By Theorem~\ref{thm: bound_quantum_functional} we have $F(f) = F^{\theta}(f)$ with $\theta = (\frac{1}{k}, \dots, \frac{1}{k})$. Moreover, from the Proposition \ref{prop: asymptotic_slice_rank_quantum} we have $\lim_{n\ra \infty} \slicerank(f^{\otimes n})^{1/n} = \min_{\theta \in \cP([k])}F^{\theta}(f) \geq F(f)$, which implies $ \lim_{n\ra \infty} \slicerank(f^{\otimes n})^{1/n} = F(f)$. This proves the claim.
\end{proof}

\section*{Acknowledgements}
OF and HT acknowledge funding from the European Research Council (ERC Grant Agreement No.\ 851716).  The research of HT is supported by the LABEX MILYON (ANR-10-LABX-0070) of Universit\'e de Lyon, within the program ``Investissements d'Avenir'' (ANR-11-IDEX-0007) operated by the French National Research Agency (ANR). 
MC acknowledges financial support from the European Research Council (ERC Grant Agreement No.~81876), VILLUM FONDEN via the QMATH Centre of Excellence (Grant No.~10059) and the Novo Nordisk Foundation (grant NNF20OC0059939 `Quantum for Life').
JZ was partially supported by a Simons Junior Fellowship and NWO Veni grant VI.Veni.212.284.

\phantomsection
\addcontentsline{toc}{section}{Bibliography}
\bibliographystyle{plainurl}
\bibliography{bibliofile}
\newpage

\appendix

\section{Representation-theoretic characterization of the moment polytope}
\label{moment polytope}

In this section we prove %
Lemma~\ref{moment_polytope_dual}. 

We recall some notions and results of geometric invariant theory and representation theory. We refer to~\cite{Ness}, \cite{Bri87},~\cite{Franz},~\cite{MichaelThesis}, and~\cite{Brgisser2019TowardsAT} for more information. 
Let $\GL(d)$ be the group of $d\times d$ invertible matrices over the complex numbers. Let $\cH$ be a complex finite-dimensional vector space, with $\dim(\cH) = d$. Denote by $M(d)$ the set of complex $d\times d$ matrices, and denote by $\Herm(d)$ the set of $d\times d$ Hermitian matrices.
We define the representation $\pi$ of $\GL(d)$ on $\cH^{\otimes k}$ by $\pi(g)f\coloneqq (g\otimes \dots \otimes g)f$ for all $g \in \GL(d)$ and $f\in \cH^{\otimes k}$. Let $\GL(d)\cdot f\coloneqq \{\pi(g)f: g \in \GL(d) \}$ denote the orbit of $f$ under the action of~$\GL(d)$.  For any nonzero vector $f \in \cH^{\otimes k}$, we define the function:
\begin{align*}
F_{f}: \quad &\GL(d) \ra \RR\\
&g \mapsto \frac{1}{2}\log\|\pi(g)f\|^{2}.
\end{align*}
The following definition defines the gradient of $F_f$ at $g = I$.    
\begin{definition}
	The \emph{moment map} is the function $\mu: \cH^{\otimes k}\setminus \{0\} \ra \Herm(d)$ defined by the property that for all $H \in \Herm(d)$ we have
	$\tr[\mu(f)H] = \partial_{t=0}F_{f}(e^{tH})$.
\end{definition}  
Let $H \in \Herm(d)$. Then $\partial_{t=0}F_{f}(e^{tH}) =  \partial_{t=0} \frac{\inner{f}{\pi(e^{tH})f}}{\|f\|^{2}}$. Therefore, we have
\begin{align*}
\tr[\mu(f)H] &=  \partial_{t=0} \frac{\inner{f}{\pi(e^{tH})f}}{\|f\|^{2}}\\
&= \frac{\inner{f}{(\sum_{j=1}^{k}I^{\otimes j-1}\otimes H \otimes I^{\otimes n-j})f}}{\|f\|^2}\\
&=\sum_{j=1}^{k}\tr\left[\frac{ff^{\dagger}}{\|f\|^2}(I^{\otimes j-1}\otimes H \otimes I^{\otimes n-j}) \right]\\
&= \sum_{j=1}^{k}\tr[\rho_j(f)H],
\end{align*}
where $\rho_{j}(f)$ denotes the $j$th reduced density matrix of $\rho(f) = \frac{ff^{\dagger}}{\|f\|^2}$. Thus, $\mu(f) = \sum_{j=1}^{k}\rho_{j}(f)$.

 Following~\cite{Fulton1991RepresentationTA}, any rational irreducible representations of $\GL(d)$ can be labeled by highest weight $\lambda \in \NN^{d}$ such that $\lambda_1 \geq \dots \geq \lambda_d$. For any natural number $n \geq 1$, consider the representation $\Pi$ of $\GL(d)$ on $(\cH^{\otimes k})^{\otimes n}$ by $\Pi(g) \cdot v\coloneqq (\pi(g)\otimes \dots \otimes \pi(g))v$ for all $v \in (\cH^{\otimes k})^{\otimes n}$. Let $V$ be a finite-dimensional rational representation of $\GL(d)$. For each highest weight $\lambda$ of $\GL(d)$, we denote by $V_{\lambda}$ the $\lambda$-isotypical component of $V$. Let $Z \subseteq V$ be a Zariski closed set. We denote by $\CC[Z]_{n}$ the degree-$n$ part of the \emph{coordinate ring} of $Z$. %
  Letting $\lambda = (\lambda_1,\dots,\lambda_d)$ be a highest weight of $\GL(d)$, we define $\lambda^{*} = (-\lambda_d,\dots,-\lambda_1)$. For any nonzero vector $f \in \cH^{\otimes k}$, the following lemma says that the moment polytope $\Delta(f)$ has another representation theoretic description.
 \begin{lemma}[{\cite{Bri87},~\cite{Franz}, \cite[Theorem 11]{Str05} or \cite[Chapter 6]{Zui18_thesis}}] 
 	Let $f \in \cH^{\otimes k}$ be nonzero. Then
 	\begin{align*}
 	\Delta(f) &= \overline{\left\{ \lambda/n: \exists n \in \NN_{\geq 1}, (\CC[\overline{\GL(d) \cdot f}]_n)_{\lambda^{*}} \neq 0 \right\}}\\
 	& = \overline{\{ \lambda/n: \exists n \in \NN_{\geq 1}, P_{\lambda} f^{\otimes n} \neq 0\}},
 	\end{align*}
 	where $P_{\lambda}$ is the projector from  $(\cH^{\otimes k})^{\otimes n}$ onto the $\lambda$-isotypical component in the decomposition of $(\cH^{\otimes k})^{\otimes n}$ with respect to $\Pi$.
 \end{lemma}

 \begin{proof}[Proof of \autoref{moment_polytope_dual}]
By Schur--Weyl duality we have a decomposition of the space $(\cH^{\otimes k})^{\otimes n}$ as
\begin{align*}
(\cH^{\otimes k})^{\otimes n} \cong \bigoplus_{\lambda \vdash_{d}kn}\mathbb{S}_{\lambda}(\cH) \otimes [\lambda].
\end{align*} 
For $\lambda \vdash_{d}kn$, let $P_{\lambda}$ be the projector onto the isotypical component of type $\lambda$, that is, onto the subspace of $(\cH^{\otimes k})^{\otimes n}$ which isomorphic to $\mathbb{S}_{\lambda}(\cH) \otimes [\lambda]$, since all irreducible representations of~$\Pi$ are labeled by the partitions of $kn$ in at most $d$ parts. Therefore,
\begin{align*}
\Delta(f) = \overline{\left\{\frac{\lambda}{n}: \exists n\geq \NN_{\geq 1}, \lambda \vdash_{d} kn, P_{\lambda}f^{\otimes n} \neq 0 \right\}},
\end{align*}
completing the proof. %
\end{proof}

\section{Sub-multiplicativity of the symmetric quantum functional}
\label{proof_quantumfunctional}

In this section we prove that the symmetric quantum functional $F$ is sub-multiplicative. For symmetric tensors this follows from \autoref{th:main2}. (In fact, \autoref{th:main2} says that the symmetric quantum functional is multiplicative on symmetric tensors.) Here we prove that the symmetric quantum functional is sub-multiplicative on arbitrary tensors (not necessarily symmetric). The argument is an adaptation of the argument in \cite{CVZ18} to the symmetric quantum functional.

\newcommand{\Kron}{\mathrm{Kron}}

\begin{lemma}\label{lem:mp-kron}
	For all tensors $s \in V^{\otimes k}$ and $t \in W^{\otimes k}$ we have $\Delta(s \otimes t) \subseteq \Delta(s) \otimes_{\Kron} \Delta(t)$ where 
	\begin{align*}
	\Delta(s) \otimes_{\Kron} \Delta(t) \coloneqq \mathrm{closure} \bigl\{ \bar{\mu}:  \bar{\lambda} \in \Delta(s), \bar{\lambda'} \in \Delta(t), P_{\mu} (P_{\lambda} \otimes P_{\lambda'}) \neq 0 \bigr\}.
	\end{align*}
\end{lemma}
\begin{proof}
	Let $\dim(V) =d$ and $\dim(W) = d'$.
	If $\bar{\mu} \in \Delta(s \otimes t)$, then for some $n$, we have $P_{\mu} (s \otimes t)^{\otimes n} \neq 0$. We have $\sum_{\lambda \vdash_{d}kn} P_{\lambda} = \Id_{V^{\otimes kn}}$ and $\sum_{\lambda'  \vdash_{d'}kn} P_{\lambda'} = \Id_{W^{\otimes kn}}$. Thus, we can write
	\begin{align*}
	P_{\mu} (s\otimes t)^{\otimes n} &= P_{\mu} \Bigl(\sum_{\lambda, \lambda'} P_{\lambda} \otimes P_{\lambda'}\Bigr)  (s\otimes t)^{\otimes n}.
	\end{align*}
	So there exists $\lambda, \lambda'$ such that $P_{\mu} (P_{\lambda} \otimes P_{\lambda'}) (s \otimes t)^{\otimes n} \neq 0$. But this implies that 
	$P_{\lambda} s^{\otimes n} \neq 0$,  %
	$P_{\lambda'} t^{\otimes n} \neq 0$, and  %
	$P_{\mu} (P_{\lambda} \otimes P_{\lambda'}) \neq 0$,
	which completes the proof.
\end{proof}
\begin{proposition}[Sub-multiplicativity of the symmetric quantum functional] 
	\label{Submultiplicativity}
	For every $s \in V^{\otimes k}$ and $t \in W^{\otimes k}$ we have
	$F(s \otimes t) \leq F(s)F(t)$.
\end{proposition}
\begin{proof}
	Let $d = \dim(V)$ and $d' = \dim(W)$.
	Let $E = \log_2 F$.
	We need to prove $	E(s \otimes t) \leq E(s)+E(t)$. By definition
	\begin{align*}
	E(s \otimes t) &= \max_{p \in \Delta(s \otimes t)} H(p)
	\leq \max_{p \in \Delta(s) \otimes_{\Kron} \Delta(t)} H(p).
	\end{align*}
	But if $p \in \Delta(s) \otimes_{\Kron} \Delta(t)$, then there exists $\mu$ a partition of $kn$ in at most $dd'$ parts such that $P_{\mu} (P_{\lambda} \otimes P_{\lambda'}) \neq 0$ with $\bar{\lambda} \in \Delta(s)$ and $\bar{\lambda'} \in \Delta(t)$ by Lemma~\ref{lem:mp-kron}.
	It is shown in \cite[Proposition 3]{CM06} that if $P_{\mu} (P_{\lambda} \otimes P_{\lambda'}) \neq 0$, then $H(\bar{\mu}) \leq H(\bar{\lambda}) + H(\bar{\lambda'})$. This shows that $E(s \otimes t) \leq E(s) + E(t)$.
\end{proof}

\end{document}